\documentclass[12pt]{article}
\usepackage{lMac}
\usepackage{becMac}

\renewcommand{\thetheorem}{\arabic{theorem}}

\newcommand{\cHu}{\mathcal{H}}
\newcommand{\Qc}{{\check Q}_-}
\newcommand{\cmp}{\mathrm{cp}}

\newcommand{\Sc}{{\check S}}
\newcommand{\Su}{S}

\def\showintremarks{n}   
\def\showissues{n}       

\begin{document}
\title{The Algebra of Block Spin Renormalization Group Transformations}

\author{Tadeusz Balaban}
\affil{\small Department of Mathematics \authorcr
       Rutgers, The State University of New Jersey \authorcr
       tbalaban@math.rutgers.edu\authorcr
       \  }

\author{Joel Feldman\thanks{Research supported in part by the Natural 
                Sciences and Engineering Research Council 
                of Canada and the Forschungsinstitut f\"ur 
                Mathematik, ETH Z\"urich.}}
\affil{Department of Mathematics \authorcr
       University of British Columbia \authorcr
       feldman@math.ubc.ca \authorcr
       http:/\hskip-3pt/www.math.ubc.ca/\squig feldman/\authorcr
       \  }

\author{Horst Kn\"orrer}
\author{Eugene Trubowitz}
\affil{Mathematik \authorcr
       ETH-Z\"urich \authorcr
       knoerrer@math.ethz.ch, trub@math.ethz.ch \authorcr
       http:/\hskip-3pt/www.math.ethz.ch/\squig knoerrer/}


\maketitle

\begin{abstract}
\noindent
Block spin renormalization group is the main tool  used in our  program 
to see symmetry breaking in a weakly interacting many Boson system 
on a three dimensional lattice at low temperature. 
In this paper, we discuss some of its purely algebraic aspects 
in an abstract setting. For example, we derive some ``well known'' identities 
like the composition rule and the relation between critical 
fields and background fields. 

\end{abstract}


\newpage
One standard implementation of the renormalization group philosophy \cite{Wil} 
uses block spin transformations. See \cite{KAD,BalLausane,GK,BalPalaiseau,Dim1}.
Concretely, suppose  we are to control a functional integral on a finite\footnote{Usually, 
the finite lattice is a ``volume cutoff'' infinite lattice and one 
wants to get bounds that are uniform in the size of the volume cutoff.} 
lattice $\cX_-$ of the form
\begin{equation}\label{introbstrbasiscfi}
\int \smprod_{x\in\cX_-} \sfrac{d\phi^*(x) d\phi(x)}{2\pi i}\,
 e^{A(\al_1,\cdots,\al_s;\phi^*,\phi)}
 \end{equation}
with an action $A(\al_1,\cdots,\al_s;\phi_*,\phi)$ that is a function 
of external complex valued fields $\al_1$, $\cdots$, $\al_s$,
and the two\footnote{ In the actions, we  treat $\phi$ and its complex conjugate $\phi^*$ as independent variables.} complex fields $\phi_*,\phi$ on $\cX_-$. This scenario occurs in \cite{PAR1,PAR2}, where  we use  block 
spin renormalization group maps to exhibit the formation of a potential 
well, signalling the onset of symmetry breaking in
a many particle system of weakly interacting Bosons in three space dimensions. 
(For an overview, see \cite{ParOv}.)
For simplicity, we suppress the external fields in this paper.

Under the  renormalization group approach to controlling integrals like
\eqref{introbstrbasiscfi} one successively ``integrates out'' 
lower and lower energy degrees of freedom. In the block spin formalism this 
is implemented by considering a  decreasing sequence of sublattices of $\cX_-$.  
The formalism produces, for each such sublattice, a representation of the
integral \eqref{introbstrbasiscfi} that is a functional integral whose
integration variables are indexed by that sublattice. To pass from the representation
associated with one sublattice $\cX\subset\cX_-$, with integration variables $\psi(x)$, $x\in\cX$, to the 
representation associated to the next coarser sublattice $\cX_+\subset\cX$,
with integration variables $\th(y)$, $y\in\cX_+$, one 
\begin{itemize}[leftmargin=*, topsep=2pt, itemsep=0pt, parsep=0pt]
\item 
paves $\cX$ by rectangles centered at the points of $\cX_+$ 
(this is illustrated in the figure below --- the dots, both small
and large, are the points of $\cX$ and the large dots are the points 
of $\cX_+$)
and then, 
\begin{figure}[ht]
  \centering
    \includegraphics[width=0.5\textwidth]{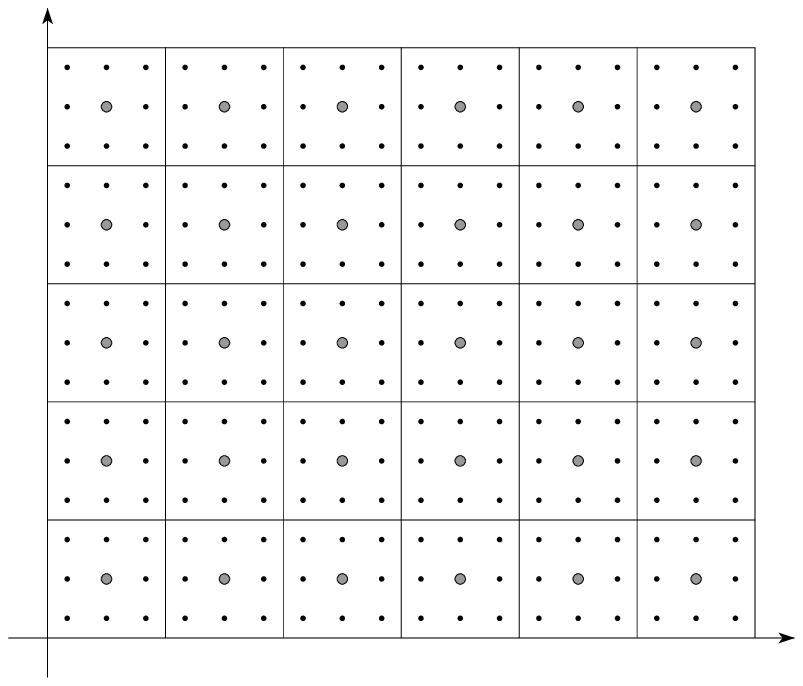}
  \caption{The lattices $\cX$ and $\cX_+$}
\end{figure}
\item 
for each $y\in\cX_+$ integrates out all values of $\psi$ whose ``average 
value'' over the rectangle centered at $y$ is equal to $\th(y)$.
The precise ``average value'' used is determined by an averaging profile.
One uses this profile to define an averaging operator $Q$ from  the space $\cH$ of fields on 
 $\cX$ to the space $\cH_+$ of fields on  $\cX_+$. One then implements the
``integrating out'' by first, inserting, into the integrand, $1$ expressed as a 
constant times the Gaussian integral
\begin{equation}\label{eqnBSgaussianOne}
\int \smprod_{y\in\cX_+} \sfrac{d\th^*(y) d\th(y)}{2\pi i}
                        e^{-b\< \th^*-Q\,\psi_*\,,\, \th-Q\,\psi) \>} 
\end{equation}
with some constant $b>0$, and then interchanging the order of the $\th$ and
$\psi$ integrals.
\end{itemize}
For example, in \cite{ParOv,PAR1,PAR2} the model is initially formulated as 
a functional integral with integration variables indexed by 
a lattice\footnote{The volume cutoff is determined by $L_\tp$ and $L_\sp$.}
$\big(\bbbz/L_\tp\bbbz\big)\times\big(\bbbz^3/L_\sp\bbbz^3\big)$.
After $n$ renormalization group steps this lattice is scaled down
to $\cX_n= \big(\sfrac{1}{L^{2n}}\bbbz \big/ \sfrac{L_\tp}{L^{2n}}\bbbz\big)
   \times \big(\sfrac{1}{L^n}\bbbz^3 \big/\sfrac{L_\sp}{L^n}\bbbz^3 \big)$. 
The decreasing family of sublattices is
$\cX_{j}^{(n-j)} 
= \big(\sfrac{1}{L^{2j}}\bbbz \big/ \sfrac{L_\tp}{L^{2n}}\bbbz\big)
   \times \big(\sfrac{1}{L^{j}}\bbbz^3 \big/
                      \sfrac{L_\sp}{L^n}\bbbz^3 \big)$, 
$j=n$, $n-1$, $\cdots$.  The abstract lattices
$\cX_-$, $\cX$, $\cX_+$ in the above framework correspond to
$\cX_n$, $\cX_0^{(n)}$ and $\cX_{-1}^{(n+1)}$, respectively.

Return to the abstract setting. The integral
is often controlled using stationary phase/steepest descent.
The contributions to the integral that come from integration variables
close to their critical values are called ``small field'' contributions.
At the end of every step, the small field contribution to the original
integral \eqref{introbstrbasiscfi} is, up to a multiplicative normalization 
constant\footnote{See Remark \ref{remBSactionRecur} for the core of the
recursion responsible for this form.},
of the form 
\begin{equation}\label{introbstrstepnfi}
\int \smprod_{x\in\cX} \sfrac{d\psi^*(x) d\psi(x)}{2\pi i}\,
e^{-\< \psi^*-Q_-\,\phi_*\,,\, \fQ(\psi-Q_-\,\phi) \> - \fA(\phi_*,\phi)
                       +\cE(\psi^*,\psi) }
                       \bigg|_{\atop{\phi_* = \phi_{*\rm bg}(\psi^*,\psi)}
                                         {\phi = \phi_{\rm bg}(\psi^*,\psi)}} 
\end{equation}
where
\begin{itemize}[leftmargin=*, topsep=2pt, itemsep=0pt, parsep=0pt]
\item
$Q_-$ is an averaging operator that maps the space $\cH_-$ of fields on 
 $\cX_-$ to the space $\cH$ of fields on  $\cX$. It is the composition of the 
averaging operations for all previous steps.
\item the exponent $\< \psi^*-Q_-\,\phi_*\,,\, \fQ(\psi-Q_-\,\phi) \> $
is a residue of the exponents in the  Gaussian integrals 
\eqref{eqnBSgaussianOne} inserted in the
previous steps. The operator\footnote{See Remark \ref{remBSactionRecur} for the
recursion relation that builds $\fQ$.} $\fQ$ is bounded and boundedly invertible
on $L^2(\cX)$.
 
\item 
the ``background fields''
\begin{equation*}
(\psi_*,\psi) \mapsto \phi_{*\rm bg}(\psi_*,\psi) \qquad
(\psi_*,\psi) \mapsto \phi_{\rm bg}(\psi_*,\psi)
\end{equation*}
map sufficiently small fields $\psi_*$, $\psi$ on  $\cX$ to fields on
$\cX_-$. They are the concatination\footnote{See Proposition
\ref{propBSconcatbackgr}.c for the recursion relation that builds
$\phi_{(*)\rm bg}$.} 
of ``steepest descent'' critical field maps 
for all previous steps.

\item
$\,\fA(\phi_*,\phi)\,$, the ``dominant part'' of the action, is
an explicit function of $\phi_*,\phi\in\cH_-$

\item
$\,\cE(\psi_*,\psi)\,$ is the contribution to the action that consists 
of ``perturbative corrections''. It is an analytic function of
$\psi_*,\psi\in\cH$.
\end{itemize}
The next block spin renormalization group step then consists of
\begin{itemize}[leftmargin=*, topsep=2pt, itemsep=0pt, parsep=0pt]

\item rewriting \eqref{introbstrstepnfi}, 
by inserting $1$ expressed as a constant times \eqref{eqnBSgaussianOne},
as
\begin{equation}\label{eqnBSintFunctInt}
\begin{split}
&\int \smprod_{y\in\cX_+} \sfrac{d\th^*(y) d\th(y)}{2\pi i}
  \int \smprod_{x\in\cX} \sfrac{d\psi^*(x) d\psi(x)}{2\pi i}\,
e^{-b\< \th^*-Q\,\psi_*\,,\, \th-Q\,\psi \>} \\ \noalign{\vskip-0.2in}
&\hskip2in
e^{-\< \psi^*-Q_-\,\phi_*\,,\, \fQ(\psi-Q_-\,\phi) \> - \fA(\phi_*,\phi)
                       +\cE(\psi^*,\psi) }
                       \bigg|_{\atop{\phi_* = \phi_{*\rm bg}(\psi^*,\psi)}
                                            {\phi = \phi_{\rm bg}(\psi^*,\psi)}} 
\end{split}
\end{equation}
up to a multiplicative normalization constant,

\item
and performing a stationary phase argument, for the $\psi$ integral, 
around appropriate critical fields\footnote{$\psi_{*\rm cr}(\th_*,\th)$
and $\psi_{\rm cr}(\th_*,\th)$ need not be complex conjugates of each other} $\psi_{*\rm cr}(\th_*,\th)$, $\psi_{\rm cr}(\th_*,\th)$ that 
map sufficiently small fields $\theta_*$, $\theta$ on  $\cX_+$ to fields 
on $\cX$. 
\end{itemize}
\vskip .3cm
In this paper, we discuss some purely algebraic aspects of the block spin renormalization group in an abstract setting. We derive some ``well known'' identities 
like, 
in Proposition \ref{propBSconcatbackgr}.c, the composition rule, and,
in Proposition \ref{propBSconcatbackgr}.a, the relation between critical 
fields and background fields, and,
in Lemma \ref{lemBSdeltaAalernew}, a formula for the dominant part of the 
action in the fluctuation integral. They are used in
        Proposition \propBGAomnibus.b,
        Proposition \propBGAomnibus.a, and
        Lemma \lemSTdeA.a of
        \cite{PAR1},
respectively.

We use the following abstract environment: 
\begin{itemize}[leftmargin=*, topsep=2pt, itemsep=0pt, parsep=0pt]
\item 
Let $H_-$, $H$, $H_+$ be finite dimensional, real vector spaces 
with positive definite symmetric bilinear forms 
$\<\,\cdot\,,\,\cdot\,\>_-$, $\<\,\cdot\,,\,\cdot\,\>$,
$\<\,\cdot\,,\,\cdot\,\>_+$. These bilinear forms extend to nondegenerate
bilinear forms on their complexifications $\cH_-$, $\cHu$, $\cH_+$.
Think of $\,H_-$, $H$ and $H_+$ as being the vector spaces 
of real valued functions on the finite lattices 
$\cX_-$, $\cX$ and $\cX_+$, respectively,
and think of the complexifications $\cH_-$, $\cHu$, $\cH_+$ as being
$L^2(\cX_-)$, $L^2(\cX)$ and $L^2(\cX_+)$ respectively.

\item
Let $d\mu_{\cH}(\phi^*,\phi)$
be the volume form on $\cH$ determined by its bilinear form. If $\cH=L^2(\cX)$, then
$
d\mu_{\cH}(\phi^*,\phi)=\smprod_{x\in \cX} 
        \sfrac{ d\phi(x)^\ast\wedge d\phi(x)}{2\pi \imath}
$.

\item
Let 
\begin{equation*}
Q_-: H_- \rightarrow H \qquad Q: H \rightarrow H_+
\end{equation*}
be linear maps.
They induce $\bbbc$ linear maps between $\cH_-$, $\cHu$, $\cH_+$ which are
denoted by the same letter. We set
\begin{equation*}
\Qc = Q\circ Q_-
\end{equation*}

\item
Fix $b>0$ and a strictly positive definite (real) symmetric linear 
operator, $\fQ$,  on $H$.

\item
Let $\fA$ be a polynomial on $\cH_-\times\cH_-$.
\end{itemize}
Set, for $\phi_*,\phi \in \cH_-$, $\psi_*,\psi \in \cHu$ and
$\th_*,\th\in \cH_+$
\begin{align*}
\cA(\psi_*,\psi;\phi_*,\phi) 
&=  \< \psi_*-Q_-\,\phi_*\,,\, \fQ(\psi-Q_-\,\phi) \> + \fA(\phi_*,\phi) \\
\cA_\eff(\th_*,\th;\psi_*,\psi;\phi_*,\phi) 
&= b \< \th_*-Q\psi_*\,,\, \th-Q\psi \>_+ + \cA(\psi_*,\psi;\phi_*,\phi)  \\
\check\cA(\th_*,\th;\phi_*,\phi) 
&= \big< \th_*-\Qc\,\phi_*\,,\,\check \fQ
                \big( \th-\Qc\,\phi\big) \big>_+ + \fA(\phi_*,\phi) 
\end{align*}
where
\begin{equation}\label{eqnBSfQrecursion}
\check \fQ= \big(\sfrac{1}{b}\bbbone_{\cH_+}+Q\,\fQ^{-1}Q^*\big)^{-1}
\end{equation}

\begin{remark}\label{remBSactionRecur}
In this setting, the action of the functional integral \eqref{introbstrstepnfi}
that appears at the beginning of the renormalization group step is
\begin{equation*}
-\< \psi^*-Q_-\,\phi_*\,,\, \fQ(\psi-Q_-\,\phi) \> - \fA(\phi_*,\phi)
                       +\cE(\psi^*,\psi) 
=-\cA(\psi^*,\psi;\phi_*,\phi) +\cE(\psi^*,\psi) 
\end{equation*}
and the action of the functional integral \eqref{eqnBSintFunctInt}
that appears in the middle of the renormalization group step is
\begin{align*}
&-b\< \th^*-Q\,\psi_*\,,\, \th-Q\,\psi \>_+
      -\< \psi^*-Q_-\,\phi_*\,,\, \fQ(\psi-Q_-\,\phi) \> - \fA(\phi_*,\phi)
                       +\cE(\psi^*,\psi) \\
&\hskip2in = -\cA_\eff(\th^*,\th;\psi^*,\psi;\phi_*,\phi)  +\cE(\psi^*,\psi)
\end{align*}
We show in Proposition \ref{propBSconcatbackgr}.b, below, that when one 
substitutes the critical $\psi$ into $\cA_\eff$ one gets $\check\cA$. 
Upon scaling (and renormalizing) $\check\cA$ becomes the $\cA$
for the beginning of the next renormalization group step. Equation 
\eqref{eqnBSfQrecursion} is the recursion relation that builds the operator
$\fQ$ in $\cA(\psi_*,\psi;\phi_*,\phi) $.
\end{remark}

\begin{remark}\label{remotherrepcheckfQ}
$\ \ \ \ \check\fQ =b \big[ \bbbone_{\cH_+} -bQ{\big(bQ^* Q+\fQ\big)}^{-1}Q^*\big]\,$
\end{remark}
\begin{proof}
Apply Lemma  \ref{lembSabstrlinalg}  with $V=\cH$, $W=\cH_+$, $q=Q$, 
$q_*=Q^*$, $f=\fQ$ and $g=b\bbbone_W$.
\end{proof}

\pagebreak[2]
\begin{definition}\label{defBSbackfld}
\ 
\begin{enumerate}[label=(\alph*), leftmargin=*]
\item
Let $\cN$ be a domain in $\cH$ which is invariant under complex conjugation.
``Background fields on $\cN$'' are maps 
$
\phi_{*\bg},\phi_\bg:\cN\times\cN\rightarrow \cH_-
$
such that, for each $(\psi_*,\psi) \in \cN\times\cN\,$, the point
$\,\big(\phi_{*\bg}(\psi_*,\psi),\,\phi_\bg(\psi_*,\psi)\big)$
is a critical point of the map
\begin{equation*}
(\phi_*,\phi) \mapsto \cA( \psi_*,\psi;\,\phi_*,\phi)
\end{equation*}
That is, it solves
\begin{equation}\label{eqnBSbckgndequ}
\begin{split}
Q_-^* \fQ\,Q_- \phi_*+\nabla_{\phi}\fA (\phi_*,\phi)  &= Q_-^* \fQ \psi_*\\
Q_-^* \fQ\,Q_- \phi+\nabla_{\phi_*}\fA(\phi,\phi) &= Q_-^* \fQ\psi
\end{split}
\end{equation}
``Formal background fields'' are formal power series 
$\phi_{*\bg}(\psi_*,\psi),\,\phi_\bg(\psi_*,\psi)$, in $(\psi_*,\psi)$
with vanishing constant terms, that solve \eqref{eqnBSbckgndequ}.

\item
Let $\,\cN_+\,$ and $\cN$ be domains in  $\cH_+$ and $\cH$, respectively,
which are invariant under complex conjugation. Let $\phi_{*\bg},\phi_\bg$
be background fields on $\cN$. ``Critical fields on $\cN_+$ with respect to
$\phi_{*\bg},\phi_\bg$'' are maps 
$
\psi_{*\crt}, \psi_{\crt}:\cN_+\times\cN_+\rightarrow \cN
$
such that, for each $ (\th_*,\th) \in \cN_+ \times \cN_+$, 
the point
$\, \big(\psi_{*\crt}(\th_*,\th), \psi_{\crt}(\th_*,\th)\big)\,$
 is a critical point for the map
\begin{equation*}
\ (\psi_*,\psi) \mapsto 
\cA_\eff(\th_*,\th;\psi_*,\psi;
        \phi_{*\bg}(\psi_*,\psi),\phi_\bg(\psi_*,\psi)) 
\end{equation*}
That is, it solves
\begin{equation}\label{eqnBScritpointequ}
\begin{split}
(bQ^* Q+\fQ)\psi_*  &= bQ^*\th_* +\fQ\,Q_-\,\phi_{*\rm bg}(\psi_*,\psi)\\
(bQ^* Q+\fQ)\psi  &= bQ^*\th +\fQ\,Q_-\,\phi_{\rm bg}(\psi_*,\psi)
\end{split}
\end{equation}
If $\phi_{*\bg},\phi_\bg$ are formal background fields, then
``formal critical fields  with respect to $\phi_{*\bg},\phi_\bg$'' are 
formal power series $\psi_{*\crt}(\th_*,\th), \psi_{\crt}(\th_*,\th)$, 
in $(\th_*,\th)$ with vanishing constant terms, that solve \eqref{eqnBScritpointequ}.

\item
Let $\cN_+$ be a domain in $\cH_+$ which is invariant under complex conjugation.
``Next scale background fields on $\cN_+$'' are maps 
$
\check\phi_{*\bg},\check\phi_\bg:\cN_+\times\cN_+\rightarrow \cH_-
$
such that, for each $(\th_*,\th) \in \cN_+\times\cN_+\,$, the point
$\,\big(\check\phi_{*\bg}(\th_*,\th),\,\check\phi_\bg(\th_*,\th)\big)$
is a critical point of the map
\begin{equation*}
(\phi_*,\phi) \mapsto \check\cA( \th_*,\th;\,\phi_*,\phi)
\end{equation*}
That is, it solves
\begin{equation}\label{eqnBSnsbckgndequ}
\begin{split}
 \Qc^* \check \fQ\,\Qc \check\phi_*
+\nabla_{\check\phi}\fA (\check\phi_*,\check\phi)
                    &= \Qc^* \check \fQ\, \th_*\\
\Qc^* \check \fQ\,\Qc\check\phi
+\nabla_{\check\phi_*}\fA(\check\phi_*,\check\phi)
                    &= \Qc^* \check \fQ\,\th
\end{split}
\end{equation} 
Formal power series $\check\phi_{*\bg}(\th_*,\th)$, $\check\phi_\bg(\th_*,\th)$, 
in $(\th_*,\th)$ with vanishing constant terms, that solve 
\eqref{eqnBSnsbckgndequ} are called 
``formal next scale background fields''.
\end{enumerate}
\end{definition}

\begin{proposition}\label{propBSconcatbackgr}
Let $\,\cN_+\,$ and $\cN$ be domains in  $\cH_+$ and $\cH$, respectively,
which are invariant under complex conjugation. Let $\phi_{*\bg},\phi_\bg$
be background fields on $\cN$ and $\psi_{*\crt},\psi_\crt$
be critical fields on $\cN_+$ with respect to $\phi_{*\bg},\phi_\bg$.
Define the composition
\begin{equation}\label{eqnBScheckphibgde}
\begin{split}
\check\phi_{*\cmp}(\th_*,\th) 
&= \phi_{*\bg}\big(\psi_{*\crt}(\th_*,\th),\psi_{\crt}(\th_*,\th)\big)
\\
\check\phi_\cmp(\th_*,\th) 
&= \phi_\bg\big(\psi_{*\crt}(\th_*,\th),\psi_{\crt}(\th_*,\th)\big)
\end{split}
\end{equation}
Then, for all  $ (\th_*,\th) \in \cN_+ \times \cN_+$,
\begin{enumerate}[label=(\alph*), leftmargin=*]
\item
$\, \big(\psi_{*\crt}(\th_*,\th), \psi_{\crt}(\th_*,\th)\big)\,$ fulfils the equations
\begin{align*}
\psi_{*\crt}(\th_*,\th)
    &={(bQ^* Q+\fQ)}^{-1} 
       \big(bQ^*\th_* +\fQ\,Q_-\,\check\phi_{*\cmp}(\th_*,\th)\big)\\
\psi_\crt(\th_*,\th)
   &={(bQ^* Q+\fQ)}^{-1} 
        \big(bQ^*\th +\fQ\,Q_-\,\check\phi_\cmp(\th_*,\th)\big)
\end{align*}

\item
The effective action
\begin{align*}
&\cA_\eff\big(\th_*,\th;\psi_{*\crt}(\th_*,\th),\psi_{\crt}(\th_*,\th); 
            \check\phi_{*\cmp}(\th_*,\th),
                       \check\phi_\cmp(\th_*,\th)\big) \\
&\hskip2.5in=\check\cA(\th_*,\th; \check\phi_{*\cmp}(\th_*,\th),
                       \check\phi_\cmp(\th_*,\th))
\end{align*}

\item
$\check\phi_{*\cmp}(\th_*,\th)\,,\, \check\phi_\cmp(\th_*,\th)$ are next 
scale background fields on $\cN_+$.

\item
For any continuous function $\,\cE(\psi_*,\psi)\,$  on $\cN\times\cN$
\begin{align*}
&\int_{\cN\times\cN} \! d\mu_{\cHu}(\psi^*,\psi) \
e^{-\cA(\psi^*,\psi; \phi_{*\bg}(\psi_*,\psi),\phi_\bg(\psi_*,\psi))
                       +\cE(\psi^*,\psi) }\\
&\hskip0.2in=b^{\dim\cH_+}
   \bigg\{\int_{\cN_+\times\cN_+}\hskip-15pt d\mu_{\cH_+}(\th^*,\th)\ 
       e^{- \check\cA(\th^*,\th; \check\phi_{*\cmp}(\th^*,\th),
\check\phi_\cmp(\th^*,\th)) }\ 
               e^{\cE(\psi_{*\crt}(\th^*,\th),\psi_{\crt}(\th^*,\th))}
                      \cF(\th^*,\th)
\\
& \hskip0.3in 
+\int_{(\cH_+\times\cH_+)\setminus (\cN_+\times\cN_+)}
\hskip -40pt d\mu_{\cH_+}(\th^*,\th)
       \int_{\cN\times\cN} \hskip -15pt d\mu_{\cHu}(\psi^*,\psi) \,
   e^{- \cA_\eff(\th^*,\th;\psi^*,\psi;  \phi_{*\bg}(\psi^*,\psi),
                                               \phi_\bg(\psi^*,\psi)) 
                         + \cE(\psi^*,\psi) }\bigg\}
\end{align*}
where the fluctuation integral
\begin{align*}
\cF(\th_*,\th)&=\int_{\cD(\th_*,\th)} d\mu_{\cHu}(\de\psi_*,\de\psi) \, 
                     e^{-\de \cA(\th_*,\th;\de\psi_*,\de\psi) }
                      e^{\de\cE(\th_*,\th;\de\psi_*,\de\psi) }
\end{align*}
Here the functions $\de\cA$ and $\de\cE$ are given by
\begin{align*}
\de \cA(\th_*,\th;\de\psi_*,\de\psi) 
&=\cA_\eff\big(\th_*,\th;\psi_*,\psi;  \phi_{*\bg}(\psi_{*},\psi), \phi_\bg(\psi_*,\psi)\big)
      \Big|^{\psi_*=\psi_{*\crt}+\de\psi_*,\ \psi=\psi_{\crt}+\de\psi}
                   _{\psi_*=\psi_{*\crt},\ \psi=\psi_{\crt}}\\
\de \cE(\th_*,\th;\de\psi_*,\de\psi) 
&=\cE\big(\psi_*,\psi\big)
      \Big|^{\psi_*=\psi_{*\crt}+\de\psi_*,\ \psi=\psi_{\crt}+\de\psi}
                   _{\psi_*=\psi_{*\crt},\ \psi=\psi_{\crt}}
\end{align*}
with $\psi_{*\crt}=\psi_{*\crt}(\th_*,\th)$, $\psi_\crt=\psi_\crt(\th_*,\th)$,
and the domain
\begin{align*}
\cD(\th_*,\th)&=\set{(\de\psi_*,\de\psi)\in\cHu\times\cHu}
               {\psi_{*\crt}(\th_*,\th)+\de\psi_*
                    =\big(\psi_{\crt}(\th_*,\th)+\de\psi\big)^*\in\cN}
\end{align*}
\end{enumerate}
\end{proposition}

\noindent The formal power series versions of parts (a), (b) and (c) of
Proposition \ref{propBSconcatbackgr} are
{   \renewcommand{\thetheorem}{\ref{propBSconcatbackgr}'}
\begin{proposition}
Let $\phi_{*\bg},\phi_\bg$ be formal background fields and
$\psi_{*\crt},\psi_\crt$ be formal critical fields with respect to
$\phi_{*\bg},\phi_\bg$.
Set\footnote{We routinely use the ``optional $*$''
notation $\al_{(*)}$ to denote ``$\al_*$ or $\al$''. The equation
``$\al_{(*)}=\be_{(*)}$'' means ``$\al_*=\be_*$ and $\al=\be$''.
}
\begin{equation}
\check\phi_{(*)\cmp}(\th_*,\th) 
= \phi_{(*)\bg}\big(\psi_{*\crt}(\th_*,\th),\psi_{\crt}(\th_*,\th)\big)
\tag{\ref{eqnBScheckphibgde}'}
\end{equation}
\begin{enumerate}[label=(\alph*), leftmargin=*]
\item
$\, \big(\psi_{*\crt}(\th_*,\th), \psi_{\crt}(\th_*,\th)\big)\,$ fulfils the equations
\begin{align*}
\psi_{(*)\crt}(\th_*,\th)
    &={(bQ^* Q+\fQ)}^{-1} 
       \big(bQ^*\th_{(*)} +\fQ\,Q_-\,\check\phi_{(*)\cmp}(\th_*,\th)\big)
\end{align*}

\item The effective action
\begin{align*}
&\cA_\eff\big(\th_*,\th;\psi_{*\crt}(\th_*,\th),\psi_{\crt}(\th_*,\th); 
            \check\phi_{*\cmp}(\th_*,\th),
                       \check\phi_\cmp(\th_*,\th)\big) \\
&\hskip2.5in=\check\cA(\th_*,\th; \check\phi_{*\cmp}(\th_*,\th),
                       \check\phi_\cmp(\th_*,\th))
\end{align*}

\item
$\check\phi_{*\cmp}(\th_*,\th)\,,\, \check\phi_\cmp(\th_*,\th)$ are 
formal next scale background fields.
\end{enumerate}
\end{proposition}
\addtocounter{theorem}{-1}
}
\noindent
The proof of these Propositions will be given after Lemma \ref{lemBSpreparation}.

\begin{remark}\label{remBSremarkonbackgroundfields}
\ 
\begin{enumerate}[label=(\alph*), leftmargin=*]
\item
Part (c) of the Proposition is often called the ``composition rule''.

\item
In applications, the domain $\cN_+$ is chosen so that the
second integral on the right hand side of the formula in part (d)
is small. In that integral either $\th$ or $\th_*$ is bounded
away from the origin (``large fields'').

\item 
As in Proposition \ref{propBSconcatbackgr}', let 
$\phi_{*\bg},\phi_\bg$ be formal background fields and 
$\psi_{*\crt},\psi_\crt$ be formal critical fields with respect to
$\phi_{*\bg},\phi_\bg$. Assume, in addition, that the equations
\eqref{eqnBSnsbckgndequ}, for the next scale background fields, 
have a unique formal power series solution, that 
we denote $\check\phi_{*\bg},\check\phi_\bg$. Then by part (c) 
of Proposition \ref{propBSconcatbackgr}',
$
\check\phi_{(*)\bg}(\th_*,\th) = \check\phi_{(*)\cmp}(\th_*,\th)
$
and, by part (a) of Proposition \ref{propBSconcatbackgr}',
\begin{align*}
\psi_{(*)\crt}(\th_*,\th)
    &={(bQ^* Q+\fQ)}^{-1} 
       \big(bQ^*\th_{(*)} +\fQ\,Q_-\,\check\phi_{(*)\bg}(\th_*,\th)\big)
\end{align*}
If, in addition, $\check\phi_{(*)\bg}(\th_*,\th)$ are analytic functions 
on some domain, then so are  $\psi_{(*)\crt}(\th_*,\th)$. 
So to construct analytical critical fields, it suffices to have
\begin{itemize}[leftmargin=*, topsep=2pt, itemsep=0pt, parsep=0pt]
\item 
uniqueness of formal power series solutions 
to the next scale background field equations
\item
existence of analytic solutions to the next scale 
background field equations
\item
formal background fields
\item
formal critical fields with respect to the formal background
fields
\end{itemize}
Lemma \ref{lemBSuniquefps}, below, provides existence and uniqueness
for formal power series solutions of the critical field equations.
\end{enumerate}
\end{remark}

\begin{lemma}\label{lemBSuniquefps}
Let $\phi_{*\bg},\phi_\bg$ be formal background fields of the form
\begin{equation*}
\phi_{(*)\bg}(\psi_*,\psi) = L_{(*)}\psi_{(*)} 
                               +  \phi_{(*)\bg}^{(\ge 2)}(\psi_*,\psi)
\end{equation*}
with $\phi_{(*)\bg}^{(\ge 2)}(\psi_*,\psi)$ being of degree at least 
two\footnote{By this we mean that each nonzero monomial
in $\phi_{(*)\bg}^{(\ge 2)}$ has degree at least two. } 
in $(\psi_*,\psi)$ and with the $L_{(*)}$'s being linear operators.
If the linear operators $bQ^*Q+\fQ - \fQ Q_-L_{(*)}$ are invertible, 
then there exist unique formal critical fields with respect to 
$\phi_{*\bg},\phi_\bg$.
\end{lemma}
\begin{proof}
 Rewrite the equations \eqref{eqnBScritpointequ} in the form
\begin{align*}
(bQ^* Q+\fQ- \fQ Q_-L_*)\psi_*  &= bQ^*\th_* 
    +\fQ\,Q_-\,\phi_{*\rm bg}^{(\ge 2)}(\psi_*,\psi)\\
(bQ^* Q+\fQ- \fQ Q_-L)\psi  &= bQ^*\th 
    +\fQ\,Q_-\,\phi_{\rm bg}^{(\ge 2)}(\psi_*,\psi)
\end{align*}
As $\psi_*$ and $\psi$ are to have vanishing constant terms, this provides
a ``lower triangular'' recursion relation for the coefficients of
$(\psi_*,\psi)$. As $\cH$ and $\cH_+$ are finite dimensional,
this recursion relation trivially generates a unique solution.
\end{proof}

The proof of Proposition \ref{propBSconcatbackgr} is based on 

\begin{lemma}\label{lemBSpreparation}
For $ \phi_*,\phi \in \cH_-$ and  $\th_*,\th \in \cH_+$ set
\begin{equation*}
\tilde\psi_{(*)}(\th_{(*)},\phi_{(*)})={(bQ^* Q+\fQ)}^{-1} 
                               \big(bQ^*\th_{(*)} +\fQ\,Q_-\,\phi_{(*)}\big)
\end{equation*}
Then
$
\check\cA\big( \th_*,\th;\,\phi_*,\phi\big)
=\cA_\eff\big(\th_*,\th; \tilde\psi_*(\th_*,\phi_*),\tilde\psi(\th,\phi);\,\phi_*,\phi\big)
$
and
\begin{equation}\label{eqnBSpreparationgrad}
\begin{split}
&(\nabla_{\phi_{(*)}} \check\cA)(\th_*,\th;\phi_*,\phi) \\
&\hskip0.5in=(\nabla_{\phi_{(*)}} \cA)\big(\tilde\psi_*(\th_*,\phi_*),\tilde\psi(\th,\phi);\phi_*,\phi\big) \\
&\hskip1in + Q_-^*\fQ\,(bQ^* Q+\fQ)^{-1}
 \big[ (\nabla_{\psi_{(*)}} \cA_\eff)\big(\th_*,\th; \tilde\psi_*(\th_*,\phi_*),\tilde\psi(\th,\phi);\,\phi_*,\phi\big)\big]
\end{split}
\end{equation}
\end{lemma}
\begin{proof}  
With the abbreviation $\,\tilde \psi_{(*)} = 
          \tilde\psi_{(*)}(\th_{(*)},\phi_{(*)})$
\begin{align*}
\th-Q\tilde\psi
&=\th - Q{(bQ^* Q+\fQ)}^{-1}\big(bQ^*\th +\fQ\,Q_-\,\phi\big)\\
&= \big[\bbbone-bQ{(bQ^* Q+\fQ)}^{-1}Q^*\big]\th -\Qc\,\phi
            +QQ_-\,\phi-Q{(bQ^* Q+\fQ)}^{-1}\fQ\,Q_-\,\phi\\
&= \big[\bbbone-bQ{(bQ^* Q+\fQ)}^{-1}Q^*\big]\th -\Qc\,\phi\\
&\hskip1in       +Q{(bQ^* Q+\fQ)}^{-1}
                  \big[(bQ^* Q+\fQ)-\fQ\big]Q_-\,\phi\\
&= \big[\bbbone-bQ{(bQ^* Q+\fQ)}^{-1}Q^*\big]
         \big(\th -\Qc\,\phi\big)\\
 \tilde\psi-Q_-\,\phi
 &=  {(bQ^* Q+\fQ)}^{-1}\big(bQ^*\th +\fQ\,Q_-\,\phi\big)
                    -Q_-\,\phi\\
 &=  {(bQ^* Q+\fQ)}^{-1}\big(bQ^*\th +\fQ\,Q_-\,\phi
                    -bQ^* QQ_-\,\phi-\fQ\,Q_-\,\phi\big)\\
 &=  b{(bQ^* Q+\fQ)}^{-1}Q^*\big(\th 
                    -\Qc\,\phi\big)
\end{align*}
Therefore
\begin{align*}
&\check\cA\big( \th_*,\th;\,\phi_*,\phi\big)
-\cA_\eff\big(\th_*,\th; \tilde\psi_*,\tilde\psi;\,\phi_*,\phi\big) \\
&\hskip 0.75cm
= \big< \th_*-\Qc\,\phi_*\,,\,\check\fQ
                      \big( \th-\Qc\,\phi\big) \big>_+
-b \big< \th_*-Q\tilde\psi_*\,,\, \th-Q\tilde\psi \big>_+\\
&\hskip 7cm
- \big< \tilde\psi_*-Q_-\,\phi_*\,,\, \fQ(\tilde\psi-Q_-\,\phi) \big>\\
&\hskip 0.75cm
= b\big< \th_*-\Qc\,\phi_*\,,\, \cO  \big( \th-\Qc\,\phi\big) \big>_+ 
\end{align*}
where, by Remark \ref{remotherrepcheckfQ},
\begin{align*}
\cO&= \big[ \bbbone -bQ{\big(bQ^* Q+\fQ\big)}^{-1}Q^*\big]
- \big[\bbbone-bQ{(bQ^* Q+\fQ)}^{-1}Q^*\big]^2\cr&\hskip0.5in
-bQ{(bQ^* Q+\fQ)}^{-1}\fQ{(bQ^* Q+\fQ)}^{-1}Q^*
\\
&= b\big[ \bbbone -bQ{\big(bQ^* Q+\fQ\big)}^{-1}Q^*\big]
         Q{(bQ^* Q+\fQ)}^{-1}Q^*\cr&\hskip0.5in
-bQ{(bQ^* Q+\fQ)}^{-1}\fQ{(bQ^* Q+\fQ)}^{-1}Q^*
\\
&= bQ\big[\bbbone -{\big(bQ^* Q+\fQ\big)}^{-1}bQ^* Q
         -{(bQ^* Q+\fQ)}^{-1}\fQ\big]
                 {(bQ^* Q+\fQ)}^{-1}Q^*\\
&=0
\end{align*}
This proves the first statement. The second follows by the chain rule 
and the observation that 
$\,\nabla_{\phi_{(*)}} \cA_\eff =\nabla_{\phi_{(*)}} \cA\,$.
\end{proof}

\begin{proof}[Proof of Propositions \ref{propBSconcatbackgr} 
            and \ref{propBSconcatbackgr}']
The proof of Proposition \ref{propBSconcatbackgr}' is virtually identical
to that of Proposition \ref{propBSconcatbackgr}.a,b,c, so we just give the 
proof of Proposition \ref{propBSconcatbackgr}.
Part (a) follows immediately from \eqref{eqnBScritpointequ} and 
\eqref{eqnBScheckphibgde}.
Now evaluate the conclusions of  Lemma \ref{lemBSpreparation} at 
$\,\phi_{(*)}= \check\phi_{(*)\cmp}(\th_*,\th)\big)\,$. 
The formula for $\check\cA$ in Lemma \ref{lemBSpreparation} directly gives 
part (b). 
The right hand side of \eqref{eqnBSpreparationgrad} vanishes upon this evaluation
by parts (a) and (b) of Definition \ref{defBSbackfld}. 
This shows that $\,\big(\check\phi_{*\cmp}(\th_*,\th)\,,\, \check\phi_\cmp(\th_*,\th)\big)\,$ is critical for the map 
$
\,(\phi_*,\phi) \mapsto \check\cA\big( \th_*,\th;\,\phi_*,\phi\big)\,
$, which proves part (c). 
Now
\begin{align*}
&b^{-\dim\cH_+}\int_{\cN\times\cN} \! d\mu_{\cHu}(\psi^*,\psi) \,e^{
     -\cA(\psi^*,\psi; \phi_{*\bg}(\psi^*,\psi),\phi_\bg(\psi^*,\psi))
                        +\cE(\psi^*,\psi) }\\
&\hskip0.1in=\int\! d\mu_{\cH_+}(\th^*,\th)
              \!\int_{\cN\times\cN} \hskip -15pt
                              d\mu_{\cHu}(\psi^*,\psi) \,
   e^{- b \< \th^*-Q\psi^*\,,\, \th-Q\psi \>_+ 
  -\cA(\psi^*,\psi;  \phi_{*\bg}(\psi^*,\psi),\phi_\bg(\psi^*,\psi))
                       \,+\,\cE(\psi^*,\psi) }\\
&\hskip0.1in=\int d\mu_{\cH_+}(\th^*,\th)
       \int_{\cN\times\cN} \hskip -12pt d\mu_{\cHu}(\psi^*,\psi) \,
   e^{- \cA_\eff(\th^*,\th;\psi^*,\psi;  \phi_{*\bg}(\psi^*,\psi),
                                               \phi_\bg(\psi^*,\psi)) 
                         + \cE(\psi^*,\psi) }\\
&\hskip0.1in= \int_{\cN_+\times\cN_+}  d\mu_{\cH_+}(\th^*,\th)
      \int_{\cN\times\cN} \hskip -12pt d\mu_{\cHu}(\psi^*,\psi) \,
   e^{- \cA_\eff(\th^*,\th;\psi^*,\psi;  \phi_{*\bg}(\psi^*,\psi),
                                               \phi_\bg(\psi^*,\psi)) 
                         + \cE(\psi^*,\psi) }
                       \\
&\hskip0.2in+\int_{\cH_+\times\cH_+\setminus \cN_+\times\cN_+}
\hskip -30pt d\mu_{\cH_+}(\th^*,\th)
       \int_{\cN\times\cN} \hskip -10pt d\mu_{\cHu}(\psi^*,\psi) \,
   e^{- \cA_\eff(\th^*,\th;\psi^*,\psi;  \phi_{*\bg}(\psi^*,\psi),
                                               \phi_\bg(\psi^*,\psi)) 
                         \,+\, \cE(\psi^*,\psi) }
\end{align*}
Making the change of variables $\psi^*=\psi_{*\crt}(\th^*,\th)+\de\psi_*$,
$\psi=\psi_\crt(\th^*,\th)+\de\psi$ in the inner integral of the upper line
and applying part (b) gives part (d).
\end{proof}

From now on we assume that the function $\fA(\phi_*,\phi)$ in the
definitions of $\cA$ and $\check\cA$ is of the form
\begin{equation}\label{eqnBSpolyAction}
\fA(\phi_*,\phi)=\<\phi_*,D\phi\>_-+P(\phi_*,\phi)
\end{equation}
where
\begin{itemize}[leftmargin=*, topsep=2pt, itemsep=0pt, parsep=0pt]
\item
$P$ is a polynomial whose nonzero monomials are each of 
degree at least two and
\item
$D$ a linear operator  on $\cH_-$ such that both
the operators $\,(D+Q_-^* \fQ\,Q_- )\,$ and 
$\,(D+\Qc^* \check \fQ\,\Qc)\,$ are invertible. We define the ``Green's
functions''
\begin{equation}\label{eqnBSdefinitionScheckS}
\Su=(D+Q_-^* \fQ\,Q_-)^{-1} \qquad \qquad \Sc=(D+\Qc^* \check \fQ\,\Qc)^{-1}
\end{equation}
\end{itemize}
We think of $D$ as a differential operator, possibly shifted
by a chemical potential.

\begin{remark}\label{remBSremarkonbackgroundfieldsB}
In this setting, the background field equations \eqref{eqnBSbckgndequ} become
\begin{equation}
\phi_{(*)}  = S^{(*)} Q_-^* \fQ \psi_{(*)} - S^{(*)}P'_{(*)} (\phi_*,\phi)
\tag{\ref{eqnBSbckgndequ}'}
\end{equation}
where $P'_*(\phi_*,\phi)= \nabla_{\phi}P (\phi_*,\phi)$ and
$P'(\phi_*,\phi)= \nabla_{\phi_*}P (\phi_*,\phi)$. Similarly, the next 
scale background field equations \eqref{eqnBSnsbckgndequ} become
\begin{equation}
\check\phi_{(*)} = \check S^{(*)}\Qc^* \check \fQ\, \th_{(*)} - 
               \check S^{(*)}P'_{(*)} (\check\phi_*,\check\phi)
\tag{\ref{eqnBSnsbckgndequ}'}
\end{equation}
\end{remark}

We now continue with our study of the critical field, following the
plan of Remark \ref{remBSremarkonbackgroundfields}.c.
To describe the leading part of the critical field,  we set
\begin{equation}\label{eqnBSdefDe}
\De = \fQ -\fQ\,Q_- \Su Q_-^* \fQ :\ \cH \longrightarrow \cH
\end{equation}
From now on we assume that $\,\De + bQ^* Q\,$ is invertible and 
define\footnote{We shall show, in Lemma \ref{lemBSdeltaAalernew},
below, that $C$ is the covariance for the fluctuation integral.} 
the ``covariance''
\begin{equation}\label{eqnBSdefCascovariance}
C=(\De + bQ^* Q)^{-1}:\ \cH \longrightarrow \cH
\end{equation}

\begin{proposition}\label{propFormalFldSlns}
Assume that in the setting \eqref{eqnBSpolyAction}, each
nonzero monomial of $P$ is of  degree at least three.
Then there exist unique formal background
fields $\phi_{(*)\bg}$ and unique formal next scale background fields
$\check\phi_{(*)\bg}$. They are of the form
\begin{align*}
\phi_{(*)\bg}(\psi_*,\psi)& = S^{(*)} Q_-^* \fQ \psi_{(*)} 
                   + \phi_{(*)\bg}^{(\ge 2)}(\psi_*,\psi) \\
\check\phi_{(*)\bg}(\th_*,\th)& = 
            \check S^{(*)} \check Q_-^* \check \fQ \th_{(*)} 
                   + \check\phi_{(*)\bg}^{(\ge 2)}(\th_*,\th)
\end{align*}
with $\phi_{(*)\bg}^{(\ge 2)}(\psi_*,\psi)$ and 
$\check\phi_{(*)\bg}^{(\ge 2)}(\th_*,\th)$ being of degree at least two.
Furthermore, there are unique formal critical fields with respect to
$\phi_{(*)\bg}$. They are of the form
\begin{align*}
\psi_{(*)\crt}(\th_*,\th)
    &={(bQ^* Q+\fQ)}^{-1} 
       \big(bQ^*\th_{(*)} +\fQ\,Q_-\,\check\phi_{(*)\bg}(\th_*,\th)\big)\\
    &=  b C^{(*)} Q^*\,\th_{(*)} +
            \psi_{(*)\crt}^{(\ge 2)}(\th_*,\th)\big)
\end{align*}
with $\psi_{(*)\crt}^{(\ge 2)}$ being of degree at least two.
\end{proposition}
\begin{proof} 
The existence, uniqueness and forms of the formal background and 
next scale background fields are proven as Lemma \ref{lemBSuniquefps} 
was proven.
The existence and uniqueness of the formal critical field now follows from
Lemma \ref{lemBSuniquefps}. The first representation of the critical fields 
follows from parts (a) and (c) of Proposition \ref{propBSconcatbackgr}'.
For the second representation, rewrite the equations \eqref{eqnBScritpointequ}
as
\begin{align*}
(bQ^* Q+\fQ)\psi_{(*)}  &= bQ^*\th_{(*)} 
                +\fQ\,Q_-\,S^{(*)} Q_-^* \fQ \psi_{(*)} 
                   + \fQ\,Q_-\,\phi_{(*)\bg}^{(\ge 2)}(\psi_*,\psi)
\end{align*}
or
\begin{align*}
 \psi_{(*)}  &= bC^{(*)}Q^*\th_{(*)} 
                   + C^{(*)}\fQ\,Q_-\,\phi_{(*)\bg}^{(\ge 2)}(\psi_*,\psi)
\end{align*}

\end{proof}

The two representations of the critical field, $\psi_{\crt}$,
given in Proposition \ref{propFormalFldSlns}, combined with the
representation of $\check\phi_{\bg}$, suggest a formula for $bCQ^*$.
In Remark \ref{remBSedA}, below, we give an algebraic proof of this 
formula, together with a number of representations for the
Green's functions, $S$ and $\check S$, and covariance $C$.
Then, in Lemma \ref{lemBSdeltaAalernew} below, we analyze the  fluctuation
integral of Proposition \ref{propBSconcatbackgr}.d  in more detail.

\begin{remark}\label{remBSedA}
Assume that $D$ is invertible.
\begin{enumerate}[label=(\alph*), leftmargin=*]
\item
$
\De = \big( \bbbone_{\cHu} +\fQ\,Q_- D^{-1} Q_-^* \big)^{-1}\fQ
    = \fQ\big( \bbbone_{\cHu} +Q_- D^{-1} Q_-^*\,\fQ \big)^{-1}
$
\item 
Let $\,R: \cH_-\rightarrow \cH\,$ and $\,R_*: \cH\rightarrow \cH_-\,$ be linear maps such that
$\,R\,D^{-1} R_* = Q_-D^{-1}Q_-^*\,$
and such that $\,D+R_*\fQ \,R\,$ is invertible.
Then
\begin{equation*}
[D+R_*\fQ \,R]^{-1} = D^{-1} -D^{-1} R_* \De\, R\, D^{-1}
\end{equation*}
In particular
\begin{equation*}
S=D^{-1} -D^{-1} Q_-^* \De\, Q_- D^{-1}
\end{equation*}

\item 
$
\Sc= \big[ S^{-1} -  Q_-^*\fQ (\fQ+bQ^*Q)^{-1}  \fQ Q_-\big]^{-1}
 = \Su + \Su Q_-^*\, \fQ\, C\, \fQ\, Q_- \Su
$
\item
$
C =  \big(bQ^* Q+\fQ\big)^{-1} 
 +(bQ^* Q+\fQ)^{-1}\, \fQ Q_-\check S Q_-^*\fQ\,(bQ^*Q+\fQ)^{-1} 
$ 
\item 
$
bC^{(*)}Q^* 
=  \big(bQ^* Q+\fQ\big)^{-1} 
\Big[bQ^* +\fQ Q_-\check S^{(*)} \check Q_-^* \check \fQ\Big] 
$
\end{enumerate}
\end{remark}
\begin{proof} (a) By Lemma \ref{lembSabstrlinalg}, with
$V=\cH_-$, $W=\cH$, $q=Q_-$, $q_*=Q_-^*$, $f=D$
and $g=\fQ$
\begin{alignat*}{3}
\big\{ \bbbone +\fQ\,Q_- D^{-1} Q_-^* \big\}^{-1}\fQ
&= \big\{ \bbbone - \fQ\,Q_-( D+Q_-^* \fQ Q_-)^{-1}Q_-^*\big\} \fQ
&\,=\De \\
\fQ\big\{ \bbbone +Q_- D^{-1} Q_-^*\,\fQ \big\}^{-1}
&= \fQ\big\{ \bbbone - Q_-( D+Q_-^* \fQ Q_-)^{-1}Q_-^*\,\fQ\big\} 
&\,=\De 
\end{alignat*}

\Item (b)  By part (a)
\begin{align*}
\big[ D+ R_*\fQ \,R\big]\,\big[ D^{-1} -D^{-1} R_* \De\, R\, D^{-1} \big]
&= \bbbone + R_* \big[ \fQ - (\bbbone +\fQ\, R\,D^{-1}R_* )\,\De\big]\,R\,D^{-1}
\\
&= \bbbone + R_* \big[ \fQ - (\bbbone +\fQ\, Q_-D^{-1}Q_-^* )\,\De\big]\,R\,D^{-1}
\\
&=\bbbone
\end{align*}

\Item (c)   By Remark \ref{remotherrepcheckfQ}
\begin{align*}
Q^* \check \fQ\,Q 
&=  b Q^*Q \big[ \bbbone -(bQ^* Q+\fQ)^{-1}bQ^*Q\big] 
\\
&=  b Q^*Q \big[ (bQ^* Q+\fQ)^{-1} (bQ^* Q+\fQ) -(bQ^* Q+\fQ)^{-1}bQ^*Q\big] 
\\
&= (\fQ+ b Q^*Q -\fQ) (bQ^* Q+\fQ)^{-1}\fQ
\\
&=   \fQ   - \fQ (\fQ+bQ^*Q)^{-1}  \fQ
\end{align*}
Therefore
\begin{align*}
S^{-1}-\check S^{-1} 
&=Q_-^*  \fQ\,Q_- - Q_-^* Q^*\check \fQ\,Q Q_-  
= Q_-^* \fQ (\fQ+bQ^*Q)^{-1}  \fQ Q_-
\end{align*}
which gives the first representation of $\check S$. 
For the proof of the second representation, first observe that,
by \eqref{eqnBSdefDe} and \eqref{eqnBSdefCascovariance},
\begin{align*}
C^{-1}(\fQ+bQ^*Q)^{-1}  
&= ( \fQ +bQ^*Q -\fQ Q_-S Q_-^*\fQ)(\fQ+bQ^*Q)^{-1}
\\
&= \bbbone - \fQ Q_- S Q_-^* \fQ (\fQ+bQ^*Q)^{-1} 
\end{align*}
so that
\begin{equation}\label{eqnBSRGCone}
C =  (\fQ+bQ^*Q)^{-1} \big\{ \bbbone -  
    \fQ\, Q_- S Q_-^*\fQ (\fQ+bQ^*Q)^{-1}\big\}^{-1}
\end{equation}
Hence, by the first representation of $\check S$,
\begin{align*}
&\big[  \Su + \Su Q_-^*\, \fQ\, C\, \fQ\, Q_- \Su\big] \check S^{-1} -\bbbone
\\
& \hskip 1cm
= \big[  \bbbone + \Su Q_-^*\, \fQ\, C\, \fQ\, Q_- \big]
\big[ \bbbone -  S Q_-^*\fQ (\fQ+bQ^*Q)^{-1}  \fQ Q_-\big] -\bbbone
\\
& \hskip 1cm
= \Su Q_-^*\, \fQ \Big[  
C \big\{\bbbone - \fQ Q_- S Q_-^* \fQ (\fQ+bQ^*Q)^{-1} \big\} - (\fQ+bQ^*Q)^{-1} 
\Big] \fQ Q_-
\\
& \hskip 1cm =0
\end{align*}
which implies the second representation of $\check S$.

\Item (d) 
By Lemma \ref{lembSabstrlinalg} with
$q= \fQ Q_-$, $q_*= Q_-^*\fQ$, $f=S^{-1}$ and $g= -(\fQ+bQ^*Q)^{-1}$ 
\begin{equation}\label{eqnBSRGCtwo}
\begin{split}
&\big\{\bbbone -  \fQ \,Q_- S Q_-^*\fQ (\fQ+bQ^*Q)^{-1}\big\}^{-1}
\\
& \hskip 2cm= \bbbone 
+ \fQ Q_- \big[ S^{-1} -  Q_-^*\fQ (\fQ+bQ^*Q)^{-1}  \fQ Q_-\big]^{-1}Q_-^*\fQ
  (\fQ+bQ^*Q)^{-1} 
\\
& \hskip 2cm= \bbbone +  \fQ Q_- \check S Q_-^*\fQ (\fQ+bQ^*Q)^{-1}
\end{split}
\end{equation}
The second equality follows by the first representation of $\check S$ 
in part (c). Substituting \eqref{eqnBSRGCtwo} into \eqref{eqnBSRGCone} 
gives the desired representation of $C$. 

\Item (e)
By Remark \ref{remotherrepcheckfQ}
\begin{align*}
\check Q_-^* \check \fQ 
&= bQ_-^* Q^* \big[ \bbbone -bQ(bQ^*Q+\fQ)^{-1} Q^*\big]
\\
&= bQ_-^*\big[   \bbbone  -bQ^*Q(bQ^*Q+\fQ)^{-1}\big]Q^*
\\
&= bQ_-^*\fQ(bQ^*Q+\fQ)^{-1}Q^*
\end{align*}
Therefore by part (d)
\begin{align*}
bC^{(*)}Q^* 
&= \big(bQ^* Q+\fQ\big)^{-1}
\big[  bQ^* + b \fQ Q_-\check S^{(*)} Q_-^*\fQ\,(bQ^*Q+\fQ)^{-1} Q^* \big]
\\
&=  \big(bQ^* Q+\fQ\big)^{-1} 
\Big[bQ^* +\fQ Q_-\check S^{(*)} \check Q_-^* \check \fQ\Big] 
\end{align*}
\end{proof}

Define, in the setting of Proposition \propBSconcatbackgr, 
$\de\phi_{(*)\rm bg}\big(\psi_{*},\psi,\de\psi_*,\de\psi\big)$ by
\refstepcounter{equation}\label{eqnBSdephibg}
\begin{equation}
\phi_{(*)\rm bg}\big(\psi_{*}+\de\psi_*,\psi+\de\psi\big)
= \phi_{(*)\rm bg}\big(\psi_{*},\psi\big)
          +  \de\phi_{(*)\rm bg}\big(\psi_{*},\psi,\de\psi_*,\de\psi\big)
\tag{\ref{eqnBSdephibg}.a}
\end{equation}
and set
\begin{equation}
\de\check\phi_{(*)\rm bg}\big(\th_{*},\th,\de\psi_*,\de\psi\big)
=\de\phi_{(*)\rm bg}\big(\psi_{*\crt}(\th_*,\th)\,,\,
              \psi_\crt(\th_*,\th)\,,\,\de\psi_*\,,\,\de\psi\big)
\tag{\ref{eqnBSdephibg}.b}
\end{equation}
With the $\check\phi_{(*)\rm bg}(\th_*,\th)$ of 
Proposition \ref{propBSconcatbackgr} and \eqref{eqnBScheckphibgde},
\begin{equation}\label{eqnBSdephicheck}
\phi_{(*)\rm bg}\big(\psi_{*\crt}(\th_*,\th)\!+\de\psi_*,
   \psi_{\crt}(\th_*,\th)\!+\de\psi\big)
= \check\phi_{(*)\rm bg}(\th_*,\th)
      +\de\check\phi_{(*)\rm bg}\big(\th_*,\th;\de\psi_*,\de\psi\big)
\end{equation}
Also define $\de{\check\phi_{(*)}}^{(+)}\big(\th_*,\th;\de\psi_*,\de\psi\big)$
by
\begin{equation}\label{eqnBSdephicheckplus}
\de\check\phi_{(*)\rm bg}\big(\th_*,\th;\de\psi_*,\de\psi\big)
=  S^{(*)} Q_-^* \fQ\,\de\psi_{(*)}
      +\! \de{\check\phi_{(*)}}^{(+)}\big(\th_*,\th;\de\psi_*,\de\psi\big)
\end{equation}

\begin{remark}\label{remBSdephieqn}
By Remark \ref{remBSremarkonbackgroundfieldsB},
the fields $\de\check\phi_{(*)\rm bg}\big(\th_*,\th,\de\psi_*,\de\psi\big)$ 
introduced in \eqref{eqnBSdephibg} obey
\begin{align*}
 \de\check\phi_{(*)\bg}
&= S^{(*)}Q_-^* \fQ\, \de\psi_{(*)}
  - S^{(*)}P'_{(*)} (\phi_*, \phi) 
      \Big|^{\phi_{(*)}=\check\phi_{(*)\bg}(\th_*,\th)+\de\check\phi_{(*)\bg}}
     _{\phi_{(*)}=\check\phi_{(*)\bg}(\th_*,\th)}
\end{align*}
In particular, if $P=0$, then $\de\check\phi_{(*)\rm bg}
={\Su^{(*)}}Q_-^* \fQ\, \de\psi_{(*)}$. This is the motivation for the
definition of $\de{\check\phi_{(*)}}^{(+)}$ in \eqref{eqnBSdephicheckplus}.
\end{remark}

\begin{lemma}\label{lemBSdeltaAalernew}
The function $\de\cA$ appearing in the exponent of the fluctuation 
integral $\cF(\th_*,\th)$ of Proposition \ref{propBSconcatbackgr}.d is
\begin{align*}
\de \cA(\th_*,\th;\de\psi_*,\de\psi) 
&= \<\de\psi_*,C^{-1}\,\de\psi\>
-\int_0^1 \!dt\  \big< \de\psi_*\,,\,\fQ\, Q_-\,
        \de{\check\phi}^{(+)}\big(\th_*,\th;t\,\de\psi_*,t\,\de\psi\big) \big> 
\\
& \hskip 3.5cm
-\int_0^1 \!dt\  \big< \fQ\, Q_-\,
          \de{\check\phi_*}^{(+)}\big(\th_*,\th;t\,\de\psi_*,t\,\de\psi\big)
\,,\, \de\psi \big> 
\end{align*}
\end{lemma}
\begin{proof}
Set
$
\cB(\psi_*,\psi)
=\cA\big(\psi_*,\psi;\phi_{*\bg}(\psi_*,\psi),\phi_\bg(\psi_*,\psi)\big)
$.
As
\begin{align*}
\big(\nabla_{\phi_*}\cA\big)\big(\psi_*,\psi;\phi_{*\bg}(\psi_*,\psi),\phi_\bg(\psi_*,\psi)\big)
=\big(\nabla_\phi\cA\big)\big(\psi_*,\psi;\phi_{*\bg}(\psi_*,\psi),\phi_\bg(\psi_*,\psi)\big)
=0
\end{align*}
we have
\begin{alignat*}{3}
\big(\nabla_{\psi_*}\cB\big)\big(\psi_*,\psi\big)
&=\big(\nabla_{\psi_*}\cA\big)\big(\psi_*,\psi;\phi_{*\bg}(\psi_*,\psi),\phi_\bg(\psi_*,\psi)\big)
&&=\fQ\big(\psi-Q_-\,\phi_\bg(\psi_*,\psi)\big)\\
\big(\nabla_\psi\cB\big)\big(\psi_*,\psi\big)
&=\big(\nabla_\psi\cA\big)\big(\psi_*,\psi;\phi_{*\bg}(\psi_*,\psi),\phi_\bg(\psi_*,\psi)\big)
&&=\fQ\big(\psi_*-Q_-\,\phi_{*\bg}(\psi_*,\psi)\big)
\end{alignat*}
Therefore
\begin{align*}
&\cB(\psi_*+\de\psi_*,\psi+\de\psi) -\cB(\psi_*,\psi) \\
& \hskip .2cm =
\int_0^1 \hskip-5ptdt\,\Big[ 
\big< \de\psi_*\,,\, (\nabla_{\psi_*}\cB)(\psi_*+t\de\psi_*,\psi+t\de\psi)\big>
+ \big<  (\nabla_{\psi}\cB)(\psi_*+t\de\psi_*,\psi+t\de\psi)\,,\,\de\psi\big> 
\Big]\\\noalign{\vskip0.05in}
& \hskip .2cm =
\int_0^1 dt \ 
\big< \de\psi_*\,,\, \fQ(\psi+t\de\psi) 
          -\fQ\,Q_-\,\phi_\bg(\psi_*+t\de\psi_*,\psi+t\de\psi)\big> \cr
\noalign{\vskip-0.05in}& \hskip 2cm+ \int_0^1 dt \ 
\big<  \fQ(\psi_*+t\de\psi_*)
          -\fQ\,Q_-\,\phi_{* \bg}(\psi_*+t\de\psi_*,\psi+t\de\psi)
   \,,\,\de\psi\big> \\ \noalign{\vskip0.05in}
& \hskip .4cm = \big< \de\psi_*, \fQ \,\de\psi \big>
   + \big< \de\psi_*, \fQ\, \psi \big> + \big< \psi_*, \fQ \,\de\psi \big> -I
\end{align*}
where
\begin{align*}
I&= 
\int_0^1 dt\ \big< \de\psi_*\,,\, 
   \fQ\,Q_-\,\phi_\bg(\psi_{* \crt}+t\de\psi_*, \psi_\crt+t\de\psi)\big> 
\\ \noalign{\vskip-0.05in}
& \hskip 1cm + \int_0^1 dt\ 
\big< \fQ\,Q_-\,\phi_{* \bg}(\psi_{* \crt}+t\de\psi_*,\psi_\crt+t\de\psi)
   \,,\,\de\psi\big>
\end{align*}
Since 
\begin{equation*}
\cA_\eff\big(\th_*,\th;\psi_*,\psi;  
        \phi_{*\bg}(\psi_{*},\psi), \phi_\bg(\psi_*,\psi)\big)
=b\<\th_*-Q\psi_*,\th-Q\psi\>_+
  +\cB(\psi_*,\psi)
\end{equation*}
we get, using Proposition \ref{propBSconcatbackgr},
\begin{align*}
\de\cA &= b\<Q\,\de\psi_*\,,\,Q\,\de\psi\>_+ 
         -b\< Q\,\de\psi_*\,,\,\th-Q\psi_\crt\>_+
         -b\<\th_*-Q\psi_{* \crt}\,,\, Q\,\de\psi\>_+ \\
& \hskip 1cm + \big< \de\psi_*\,,\, \fQ \,\de\psi \big>
   + \big< \de\psi_*\,,\, \fQ\, \psi_\crt \big> 
   + \big< \psi_{* \crt}\,,\, \fQ \,\de\psi \big> 
   -I \displaybreak[0]\\ \noalign{\vskip0.05in}
& = \<\de\psi_*\,,\,(bQ^* Q + \fQ)\,\de\psi\>
+ \< \de\psi_*\,,\,(bQ^* Q + \fQ)\psi_\crt -bQ^* \th\> \\
& \hskip 3.9cm
+ \< (bQ^* Q + \fQ)\psi_{* \crt} -bQ^* \th_*\,,\, \de\psi\>
-I
\displaybreak[0]\\ \noalign{\vskip0.05in}
& = \<\de\psi_*\,,\,(bQ^* Q + \fQ)\,\de\psi\>
+ \< \de\psi_*\,,\,\fQ\, Q_- \check\phi_\bg \> 
+ \< \fQ\, Q_-\check \phi_{* \bg} \,,\, \de\psi\> -I
\displaybreak[0]\\ \noalign{\vskip0.05in}
& = \<\de\psi_*\,,\,(bQ^* Q\!+\!\fQ)\,\de\psi\>
-\int_0^1 \hskip-5pt dt\  \big< \de\psi_*\,,\,\fQ\, Q_-
\big[\phi_\bg(\psi_{* \crt}+t\de\psi_*,\psi_\crt+t\de\psi) - \check\phi_\bg\big] \big> 
\\
& \hskip 3.6cm
-\int_0^1 \!dt\  \big< \fQ\, Q_-
\big[\phi_{*\bg}(\psi_{* \crt}+t\de\psi_*,\psi_\crt+t\de\psi) - \check\phi_{*\bg}\big]
\,,\, \de\psi \big> 
\\ \noalign{\vskip0.05in}
& = \<\de\psi_*,(bQ^* Q\!+\! \fQ\!-\!\fQ Q_-\Su Q_-^* \fQ)\,\de\psi\>
-\!\int_0^1 \hskip-5pt dt\,  \big< \de\psi_*,\fQ\, Q_-\,
\de{\check\phi}^{(+)}\big(\th_*,\th;t\de\psi_*,t\de\psi\big) \big> 
\\
& \hskip 6cm
-\int_0^1 \!dt\,  \big< \fQ\, Q_-\,
\de{\check\phi_*}^{(+)}\big(\th_*,\th;t\de\psi_*,t\de\psi\big)
,\, \de\psi \big> 
\end{align*}
By the definition of $C$ in \eqref{eqnBSdefCascovariance}, this is 
the desired representation.
\end{proof}

\bigskip
In the course of the arguments above the following simple algebraic 
observation was used several times.
\begin{lemma}\label{lembSabstrlinalg}
Let $V$ and $W$ be vector spaces and let $\,q:V\rightarrow W\,$, $\,q_*:W\rightarrow V\,$,
$\,f:V\rightarrow V\,$ and $\,g:W\rightarrow W\,$ be linear maps. Assume that
$f$ and $\,f+q_*g\,q\,$ are invertible.
Then $\,\bbbone_W +gq\,f^{-1}q_*\,$ and $\,\bbbone_W +q\,f^{-1}q_*g\,$
are also invertible and
\begin{align*}
\big( \bbbone_W +gq\,f^{-1}q_* \big)^{-1} &= \bbbone_W -gq (f+q_*gq)^{-1} q_*\\
\big( \bbbone_W +q\,f^{-1}q_*g \big)^{-1} &= \bbbone_W -q (f+q_*gq)^{-1} q_*g
\end{align*}
\end{lemma}
\begin{proof} 
Replacing $q$ by $gq$ for the first line and $q_*$ by $q_*g$ for the second, 
we may assume that $g=\bbbone_W$. Write $\bbbone_W=\bbbone$. Then
\begin{align*}
\big(\bbbone -q (f+q_*q)^{-1} q_*\big) \big( \bbbone +qf^{-1}q_* \big)
&= \bbbone +q \big[\bbbone -(f+q_*q)^{-1}f - (f+q_*q)^{-1}q_*q\big] f^{-1} q_*
\\
&= \bbbone 
\end{align*}
and similarly
$\,
\big( \bbbone +qf^{-1}q_* \big) \big(\bbbone -q (f+q_*q)^{-1} q_*\big) =\bbbone
\,$.
\end{proof}

\newpage
\bibliographystyle{plain}
\bibliography{refs}

\begin{thebibliography}{1}

\bibitem{BalLausane}
T.~Balaban.
\newblock {The Ultraviolet Stability Bounds for Some Lattice $\si$--Models and
  Lattice Higgs--Kibble Models}.
\newblock In {\em {Proc. of the International Conference on Mathematical
  Physics, Lausanne, 1979}}, pages 237--240. Springer, 1980.

\bibitem{BalPalaiseau}
T.~Balaban.
\newblock A low temperature expansion and {``}spin wave picture{''} for
  classical {$N$}-vector models.
\newblock In {\em {Constructive Physics (Palaiseau, 1994), Lecture Notes in
  Physics, 446}}, pages 201--218. Springer, 1995.

\bibitem{ParOv}
T.~Balaban, J.~Feldman, H.~Kn{\"o}rrer, and E.~Trubowitz.
\newblock {Complex Bosonic Many--body Models: Overview of the Small Field
  Parabolic Flow}.
\newblock Preprint, 2016.

\bibitem{PAR1}
T.~Balaban, J.~Feldman, H.~Kn{\"o}rrer, and E.~Trubowitz.
\newblock {The Small Field Parabolic Flow for Bosonic Many--body Models: Part 1
  --- Main Results and Algebra}.
\newblock Preprint, 2016.

\bibitem{PAR2}
T.~Balaban, J.~Feldman, H.~Kn{\"o}rrer, and E.~Trubowitz.
\newblock {The Small Field Parabolic Flow for Bosonic Many--body Models: Part 2
  --- Fluctuation Integral and Renormalization}.
\newblock Preprint, 2016.

\bibitem{Dim1}
J.~Dimock.
\newblock {The renormalization group according to Balaban -- I. small fields}.
\newblock {\em Reviews in Mathematical Physics}, 25:1--64, 2013.

\bibitem{GK}
K.~Gawedzki and A.~Kupiainen.
\newblock {A rigorous block spin approach to massless lattice theories}.
\newblock {\em Comm. Math. Phys.}, 77:31--64, 1980.

\bibitem{KAD}
L.P. Kadanoff.
\newblock {Scaling laws for Ising models near $T_c$}.
\newblock {\em Physics}, 2:263, 1966.

\bibitem{Wil}
K.G. Wilson.
\newblock {The renormalization group: critical phenomena and the Kondo
  problem}.
\newblock {\em Rev. Mod. Phys.}, 47:773, 1975.

\end{thebibliography}

\end{document}